\documentclass{article}
\usepackage{amsmath}
\usepackage{amsfonts}
\usepackage{amscd}
\usepackage{theorem}
\usepackage{amssymb}
\usepackage{latexsym}
\usepackage[pdftex]{graphicx}

\newcommand{\C}{\mathbb{C}}

\newcommand{\oph}{\overline{\phi}}

\newcommand{\opsi}{\overline{\psi}}
\newcommand{\R}{\mathbb{R}}

\newcommand{\zG}{\Gamma}
\newcommand{\zl}{\lambda}

\newcommand{\zO}{\Omega}

\newcommand{\deff}{\stackrel{\text{\tiny def}}{=}}

\newtheorem{theorem}{\bf Theorem}
\numberwithin{theorem}{section}

\newtheorem{definition}[theorem]{\bf Definition}

\newenvironment{proof}
{\begin{trivlist}\item[\hskip\labelsep\quad{\bf Proof.}
\hspace{0.5 em}]}{\hfill \rule{0.5em}{0.5em} \end{trivlist}}

\numberwithin{equation}{section}
\numberwithin{figure}{section}

\author{Francesco Demontis \footnote{Dip. Matematica, Universit\`a di Cagliari,
Viale Merello 92, 09121 Cagliari, Italy}\,\,\footnote{Research
supported by RAS under grant PO Sardegna 2007-2013, L.R. 7/2007}		}
\title{Exact solutions to the modified \\Korteweg-de Vries equation.}

\begin{document}
\date{}
\maketitle
\thispagestyle{plain}

\begin{abstract}
A formula for certain exact solutions to the modified
Korteweg-de Vries (mKdV) equation is obtained via
the inverse scattering transform method. 
The kernel of the
relevant Marchenko integral equation is 
written with the help
of matrix exponentials as 
$$\zO(x+y;t)=Ce^{-(x+y)A}e^{8A^3 t}B,$$
where the real matrix triplet $\left(A,B,C\right)$ consists of a
constant $p\times p$ matrix $A$
with eigenvalues having positive real parts, 
a constant $p\times 1$ matrix
$B$, and a constant $1\times p$ matrix $C$ for a 
positive integer $p$. 
Using separation of variables, 
the Marchenko integral equation
is explicitly solved yielding exact solutions to 
the mKdV equation.
These solutions are constructed in terms 
of the unique solution $P$ to the Sylvester 
equation  $AP+PA=BC$ or in terms of 
the unique solutions
$Q$ and $N$ to the respective Lyapunov equations
$A^\dagger Q+QA=C^\dagger C$ and
$AN+NA^\dagger=BB^\dagger$, where the $\dagger$ denotes 
the matrix conjugate transpose. Two interesting examples
are provided.
\end{abstract}

\section{Introduction}\label{sec:1}
Consider the focusing modified Korteweg-de Vries (mKdV) equation
\begin{equation}\label{0.1}
u_t+u_{xxx}+6|u|^2u_x=0\,,
\end{equation}
where the subscripts denote the appropriate partial derivates, $u$ 
represents a real scalar function and $(x,\,t)\in \R^2$.

The modified Korteweg-de Vries (mKdV) equation arises in applications to the
dynamics of thin elastic rods \cite{MT}, phonons in anharmonic lattices
\cite{Ono}, meandering ocean jets \cite{RP}, traffic congestion \cite{KS, Nag,
LL, GDXD}, hyperbolic surfaces \cite{Sc}, ion acoustic solitons \cite{Lon},
Alfv\'en waves in collisionless plasmas \cite{KEC}, slag-metallic bath
interfaces \cite{AC}, and Schott\-ky barrier transmission lines \cite{ZDSL}.

In this paper we present a method to construct certain exact solutions to
\eqref{0.1} that are globally analytic on the entire $xt$-plane and  
decay exponentially as $x\to\pm\infty$ for each fixed $t\in\R$. The method used 
to obtain these solutions is based on the Inverse Scattering Transform
(IST) \cite{AblClar}-\cite{APT}, \cite{Lamb} and \cite{Nov}. The IST matches eq.\eqref{0.1}
to the first order system of ordinary differential equations
\begin{subequations}\label{0.2}
\begin{align}
\dfrac{d\xi}{dx} &=-i\lambda \xi+u(x,t)\,\eta,\\
\dfrac{d\eta}{dx} &=i\lambda \eta-u(x,t)\,\xi,
\end{align}
\end{subequations}
known as the Zakharov-Shabat system \cite{ZS}. In Section \ref{sec:2} we develop
its direct and the inverse scattering theory. More precisely, 
representing the corresponding scattering data 
in terms of a matrix realization \cite{AkVan6}-\cite{TUIOCOR}, the Marchenko
integral equation is separable and hence, can be solved algebraically. 
Its solution is easily related to the solution of \eqref{0.1}. 

The method used in this article has several advantages:
\begin{enumerate}
\item[1.] It is generalizable to the matrix version of eq.\eqref{0.1}
and to other (matrix) nonlinear evolution equations (see, e.g., \cite{DM2}).
\item[2.] The explicit formulas found in this paper are expressed in a concise form 
in terms of a triplet $\left(A,\,B,\,C\right)$
where $A$ is a real square matrix of dimension $p$, $C$  is a real row 
vector and $B$ is a real column vector.
When the matrix order of $A$ is very large, we have 
an explicit formula for the solution of the initial value problem for 
\eqref{0.1}. Using computer algebra, we can ``unzip" the solution in terms of
exponential, trigonometric, and polynomial functions of 
$x$ and $t$, but this unzipped expression
may take several pages!
\item[3.] Our method easily treats nonsimple bound-state poles and the time evolution
of the corresponding bound-state norming constants.  
\end{enumerate}
This paper is organized as follows: In Section \ref{sec:2} we develop the direct and the inverse
scattering theory for system \eqref{0.2} and describe how the IST allows us to get the solution to 
\eqref{0.1}. In Section \ref{sec:3} we construct the exact solutions to the initial value 
problem \eqref{0.1} in terms of real matrices $\left(A,\,B,\,C\right)$ solving 
the Marchenko equation. In Section \ref{sec:4} we write the triplet $\left(A,\,B,\,C\right)$ 
in a ``canonical" form. Finally, in Section \ref{sec:5} 
we discuss the one-soliton and a multipole solution as examples.

\section{IST method for the mKdV equation}\label{sec:2}
The mKdV equation \eqref{0.1} is solvable by the inverse scattering tranform (IST).
This method associates \eqref{0.1} with the first order system of ODE \eqref{0.2} where 
the coefficient $u$ is called the potential. We suppose that $u$ is a real scalar function
belonging to $L^1(\R)$. In this section we recall the basic ideas 
behind the IST and introduce the scattering coefficients and  Marchenko integral equations.
A more detailed exposition can be found 
in \cite{AblClar}-\cite{APT}, \cite{Lamb} and \cite{Nov}.

Let us introduce the {\it Jost functions} from the right $\opsi(\zl,x)$ and $\psi(\zl,x)$,
the Jost functions from the left $\phi(\zl,x)$ and $\oph(\zl,x)$, and the Jost 
matrix solutions $\Psi(\zl,x)$ and $\Phi(\zl,x)$ from the right and from the left 
as those solutions of \eqref{0.2} satysfying the asymptotic conditions: 
\begin{align*}
\Psi(\zl,x)&=\begin{pmatrix}\opsi(\zl,x)&\psi(\zl,x)\end{pmatrix}=\begin{cases}
e^{-i\zl Jx}[I_2+0(1)],&x\to+\infty,\\ e^{-i\zl Jx}a_l(\zl)+o(1),
&x\to-\infty,\end{cases}\\
\Phi(\zl,x)&=\begin{pmatrix}\phi(\zl,x)&\oph(\zl,x)\end{pmatrix}=\begin{cases}
e^{-i\zl Jx}[I_2+o(1)],&x\to-\infty,\\ e^{-i\zl Jx}a_r(\zl)+o(1),
&x\to+\infty,
\end{cases}
\end{align*}
where $J=\begin{pmatrix}1 & 0\\ 0 & -1\end{pmatrix}$. Here we have not written the variable $t$. 
Since this system is first order,
we have $$\Phi(\zl,x)=\Psi(\zl,x)a_r(\zl),\qquad\Psi(\zl,x)=\Phi(\zl,x)a_l(\zl),$$ where 
$a_r(\zl)$ and $a_l(\zl)$ are called {\it transmission coefficient matrices}. It is easy to prove that 
$a_r(\zl)$ and $a_l(\zl)$, for all $\zl\in\R$, are unitary matrices with unit determinant, one
being the inverse of the other (see, e.g., \cite{io, APT, V04}). It is convenient to write these 
matrices in the form 
$$a_l(\zl)=\begin{pmatrix}a_{l1}(\zl)&a_{l2}(\zl)\\ a_{l3}(\zl)&a_{l4}(\zl)
\end{pmatrix},\qquad a_r(\zl)=\begin{pmatrix}a_{r1}(\zl)&a_{r2}(\zl)\\
a_{r3}(\zl)&a_{r4}(\zl)\end{pmatrix},$$
where $a_{lk}(\zl)$ and $a_{rk}(\zl)$ for $k=1,2,3,4$ are real scalar functions. 
The Jost functions $\phi(\zl,x)$ and $\psi(\zl,x)$ \cite{io, APT, V04} are, for each 
$x\in\R$, analytic in $\C^+$, continuous in $\zl\in\overline{\C^+}$ and tend to $1$ as $|\zl|\to\infty$
from within $\C^+$, where we use $\C^+$ ($\C^-$) to denote the upper (lower) complex open half planes,
 and $\overline{\C^{\pm}}=\C^{\pm}\cup\R$. Instead, the Jost 
functions $\oph(\zl,x)$ and $\opsi(\zl,x)$ are, for each 
$x\in\R$, analytic in $\C^-$, continuous in $\zl\in\overline{\C^-}$ and tend to $1$ as $|\zl|\to\infty$
from within $\C^-$. So, it is natural to consider the $2\times 2$ matrices of functions
$$F_+(\zl,x)=\begin{pmatrix}\phi(\zl,x)&\psi(\zl,x)\end{pmatrix},\qquad
 F_-(\zl,x)=\begin{pmatrix}\opsi(\zl,x)&\oph(\zl,x)\end{pmatrix}.$$
As a result, $F_+(\zl,x)$ is, for each $x\in\R$, analytic in $\C^+$ and continuous in $\overline{\C^+}$.
Instead, $F_{-}(\zl,x)$ is, for each $x\in\R$, analytic in $\C^-$ and continuous in $\overline{\C^-}$. 
Using this information we arrive at the Riemann-Hilbert problem
\begin{equation}\label{2.1}
 F_-(\zl,x)=F_+(\zl,x)JS(\zl)J,
\end{equation}
where $$S(\zl)=\begin{pmatrix}T_r(\zl)&L(\zl)\\ R(\zl)&T_l(\zl)\end{pmatrix}$$
is called the {\it scattering matrix}. 

The direct problem can be formulated as follows: Starting from the potential $u(x)$, 
construct the scattering matrix
or, equivalently, determine one reflection coefficient $R(\zl)$ or $L(\zl)$, 
the poles (bound-states) $\zl_j$
of the transmission coefficient $T_r(\zl)$ or $T_l(\zl)$ 
(see Theorem $3.16$ of \cite{io}),
and the norming constants $c_{js}$ corresponding to those poles. 
Moreover, under the technical assumption that $a_{l1}(\zl)$,
$a_{r1}(\zl)$, $a_{l4}(\zl)$ and $a_{r4}(\zl)$ are 
invertible for each $\zl\in\R$, 
we can relate the transmission 
coefficients to the coefficients appearing in the scattering matrix $S(\zl)$. 
More precisely, we have
\begin{alignat*}{3}
T_r(\zl)&=a_{r1}(\zl)^{-1},\,\qquad R(\zl)&=-a_{l4}(\zl)^{-1}a_{l3}(\zl)=a_{r3}(\zl)a_{r1}(\zl)^{-1},\\
T_l(\zl)&=a_{l4}(\zl)^{-1},\,\qquad
L(\zl)&=-a_{r1}(\zl)^{-1}a_{r2}(\zl)=a_{l2}(\zl)a_{l4}(\zl)^{-1}.
\end{alignat*} 

The inverse scattering problem consists of the (re)-construction of the unique 
potential $u(x)$ if one knows 
the scattering data. 
In this paper, following \cite{AkVan6, AktDV7}, we solve this problem via the 
Marchenko method (see \cite{Nov, ZS}) as follows:
\begin{enumerate}
\item[a.] From the scattering data $\Big\{R(\zl),\,\left\{\zl_j,\, 
\{c_{js}\}_{s=0}^{n_j-1}\right\}_{j=m+1}^{m+n}\Big\}$, construct the function 
\begin{equation}\label{2.2}
\zO(y)\deff\hat{R}(y)+\sum_{j=m+1}^{m+n}\sum_{s=0}^{n_j-1}c_{js}\frac{y^s}{s!}e^{i\zl_j y}
\end{equation}
where $\hat{R}(y)=\frac{1}{2\pi}\int_0^\infty R(\zl)e^{i\zl y}d\zl$ 
represents the Fourier transform of $R(\zl)$.
\item[b.] Solve the Marchenko integral equation
\begin{equation}\label{2.3}
K(x,y)-\Omega(x+y)^\dagger+\int_x^\infty dz\int_x^\infty ds \,
K(x,z)\,\Omega(z+s)\,\Omega(s+y)^\dagger=0,
\end{equation}
where $y>x.$
\item[c.] Construct the potential $u(x)$ by using the following formula:
\begin{equation}\label{2.4}
u(x)=-2K(x,x).
\end{equation}
\end{enumerate}
Having presented the direct and inverse scattering problems corresponding 
to the LODE associated
to the mKdV equation, we now discuss how the IST allows us 
to obtain the solution 
to the initial value problem for \eqref{0.1}. The following diagram 
illustrates the IST procedure:
$$\begin{CD}
\fbox{$\begin{matrix}\text{given}\\ u(x,0)\end{matrix}$}@>{\begin{smallmatrix}\text{direct\
scattering\ problem}\\ \text{with\ potential}\
u(x,0)\end{smallmatrix}}>> 
\fbox{$\begin{matrix} R(\zl),\,\{\zl_j,\, 
\{c_{js}(0)\}_{s=0}^{n_j-1}\}_{j=m+1}^{m+n}\end{matrix}$}\\
@VV{\begin{smallmatrix}\text{mKdV}\\ \text{solution}\end{smallmatrix}}V @V{\begin{smallmatrix}\text{time evolution \ of}\\
\text{scattering\ data}\end{smallmatrix}}VV\\
\fbox{$u(x,t)$}@<<{\begin{smallmatrix}\text{inverse\ scattering\ problem}\\
\text{with\ time evolved\ scattering\ data}\end{smallmatrix}}<
\fbox{$R(\zl,t),\,\{\zl_j,\, 
\{c_{js}(t)\}_{s=0}^{n_j-1}\}_{j=m+1}^{m+n}$}
\end{CD}$$
To solve the initial value problem for \eqref{0.1}, we use the initial 
condition $u(x,0)$ as 
a potential in the system \eqref{0.2}. After that we develop the 
direct scattering theory
as shown above and build the scattering matrix. Successively, 
let the initial 
scattering data evolve in time. The transmission coefficient
does not change in time and, as a consequence, also the bound states do not.
The reflection 
coefficient $R(\zl)$ evolves according to $R(\zl,t)=e^{8\zl^3t}R(\zl)$. It remains 
to determine the evolution of
the norming constants. Extending our previous results on the 
nonlinear Schr\"odinger equation
(\cite{AktDV7, DM2, io, Bs}) to the mKdV equation, we obtain 
the following time evolution of
the norming constants: 
$$\begin{pmatrix} c_{jn_j-1}(t)& \ldots & c_{j0}(t)\end{pmatrix}=
\begin{pmatrix}c_{jn_j-1}(0) & \ldots & c_{j0}(0)\end{pmatrix}\,e^{-A_j^{3}t},$$ 
where $A_j$ is the matrix defined by eq. \eqref{4.7}.
Finally, we solve the inverse scattering problem by the Marchenko method 
for \eqref{0.2} with the time evolved scattering 
data $\Big\{R(\zl,t),\,\left\{\zl_j,\, 
\{c_{js}(t)\}_{s=0}^{n_j-1}\right\}_{j=m+1}^{m+n}\Big\}$,
replacing in \eqref{2.2}
$R(\zl)$ with $R(\zl,t)$ and $c_{js}$ with $c_{js}(t)$. Then the function
\begin{equation}\label{2.5}
u(x,t)=-2K(x,x;t)\,,
\end{equation} 
is a solution to the mKdV equation.

\section{Explicit solutions for the mKdV equation}\label{sec:3}
In this section we obtain two different but equivalent formulas yielding 
solutions to the mKdV equation.
Using the expressions 
of the evolved reflection coefficient and the evolved 
norming constants in \eqref{2.2}, we get  
\begin{equation}\label{3.1}
\zO(y; t)=\frac{1}{2\pi}\int_0^\infty R(\zl)e^{8\zl^3 t}e^{i\zl y}d\zl+
\sum_{j=m+1}^{m+n}\sum_{s=0}^{n_j-1}c_{js}(t)\frac{y^s}{s!}e^{i\zl_j y}\,,
\end{equation}
which satisfies the first order PDE
\begin{equation}\label{3.2}
\zO_t(y;t)+8\zO_{yyy}(y;t)=0\,.
\end{equation}

Now, following the approach of \cite{AkVan6, AktDV7, TUIOCOR}, 
we write the kernel $\zO(y)$ introduced in \eqref{2.2} in the form 
\begin{equation}\label{3.3}
\zO(y)=Ce^{-Ay}B,\quad y\geq 0\,,
\end{equation}
where $A,\,B,\,C$ are real matrices of size $p\times p,\,p\times 1$
and $1\times p$, respectively, for some integer $p$. Eq. \eqref{3.2}
suggests us to take $\zO(y;t)$ as
\begin{equation}\label{3.4}
\zO(y;t)=Ce^{-Ay}e^{8A^3 t}B,\,\,y\geq 0\,.
\end{equation}
For reasons to be clarified later, we have some further requirements on the 
triplet $\left(A,\,B,\,C\right)$. More precisely,
\begin{enumerate}
\item Our triplets $\left(A,\,B,\,C\right)$ is a {\it minimal representation} 
of the kernel $\zO(y)$, i.e,
$$\bigcap_{r=1}^{+\infty}\,\ker{CA^{r-1}}=
\bigcap_{r=1}^{+\infty}\,\ker{B^\dagger (A^\dagger)^{r-1}}=\{0\},$$ 
and we refer the reader to \cite{BaGoKa} for many details on this subject. 
\item All of the eigenvalues of the matrix $A$ have positive real parts.
\end{enumerate}
Following the procedure described in the preceding section, we find explicit solutions of 
the mKdV equation solving the Marchenko integral equation \eqref{2.3}, 
where the kernel is given by \eqref{3.4} and the unknown function $K$  
depends on $t$. 
It is immediate to calculate
\begin{equation}\label{3.5}
\zO(y;t)^\dagger=B^\dagger e^{-A^\dagger y}e^{8(A^\dagger)^3 t}C^\dagger,\,\,y\geq 0\,.
\end{equation}
By using \eqref{3.4} and \eqref{3.5}, eq. \eqref{2.3} becomes
\begin{align}\label{3.6} 
\nonumber K(x,y;t)-\Big(B^\dagger e^{-A^\dagger x} & - \int_x^\infty dz\,\int_x^\infty ds \,
K(x,z;t)\,Ce^{-Az+8A^3 t}e^{-As}\,BB^\dagger e^{-A^\dagger s}\Big)\\ 
& \cdot e^{-A^\dagger y+8(A^\dagger)^3 t}C^\dagger=0\,, \quad y>x\,.
\end{align}
Looking for a solution of \eqref{3.6} in the form 
\begin{equation}\label{3.7}
K(x,y;t)=H(x,t)e^{-A^\dagger y+8(A^\dagger)^3 t}C^\dagger\,,
\end{equation}
and introducing the matrices $Q$ and $N$ as
\begin{equation}\label{3.8}
Q=\int_0^\infty ds\,e^{-A^\dagger s}C^\dagger Ce^{-A s},\qquad
N=\int_0^\infty dr\,e^{-A r}BB^\dagger e^{-A^\dagger r}\, ,
\end{equation}
after some easy calculations we obtain
\begin{equation}\label{3.9}
H(x,t)\zG(x,t)=B^\dagger e^{-A^\dagger x}\,,
\end{equation}
where (denoting by $I_p$ the identity matrix of order $p$)
\begin{equation}\label{3.10}
\zG(x,t)=I_p+e^{-A^\dagger x+8(A^\dagger)^3t}Qe^{-2Ax+8A^3t}Ne^{-A^\dagger x}\,. 
\end{equation}
Substituting \eqref{3.9} into \eqref{3.7} we get
$$K(x,y;t)=B^\dagger e^{-A^\dagger x}\zG(x,t)^{-1}
e^{-A^\dagger y+8(A^\dagger)^3 t}C^\dagger=
B^\dagger F^{-1}(x,t)e^{-A^\dagger (y-x)}C^\dagger\,.$$
Putting
\begin{equation}\label{3.11}
F(x,t)=e^{2A^\dagger x-8(A^\dagger)^3t}+Qe^{-2Ax+8A^3t}N\,, 
\end{equation}
we arrive at the solution formula
\begin{equation}\label{3.12}
u(x,t)=-2 B^\dagger F^{-1}(x,t)C^\dagger\,.
\end{equation}
The matrices $Q,\, N,\, F(x,\,t)$ and the scalar function
$u(x,\,t)$ introduced by \eqref{3.8}, \eqref{3.11} and 
\eqref{3.12}, respectively, 
satisfy the properties 
stated in the following
\begin{theorem}\label{th:3.1} 
Suppose that the triplet $\left(A,\,B,\,C\right)$ is real and
is a minimal representation of the kernel 
$\zO(y;t)$, and that the eigenvalues of $A$ have positive real parts. Then
\begin{enumerate}
\item[a)] The matrices $Q$ and $N$ are real and positive selfadjoint, i.e $Q^\dagger=Q$ and $N^\dagger=N$.
\item[b)] The matrices $Q$ and $N$ are simultaneously invertible.
\item[c)] The matrix $F(x,t)$ is invertible on the entire $xt$-plane. Moreover, for each fixed $t$, 
$F(x,t)^{-1}\to 0$ as $x\to\pm \infty$.
\item[d)] The real scalar function $u(x,t)$ satisfies \eqref{0.1} everywhere on the $xt$-plane.
Moreover, $u(x,t)$ is analytic on the entire $xt$-plane and decays exponentially for each fixed $t$ as 
$x\to\pm \infty$.
\end{enumerate}
\end{theorem}
\begin{proof}
The proof of the items $a),\,b),\,c)$ and of the analyticity and asymptotic behaviour 
of the solution $u(x,\,t)$ is identical to the proof of items (ii) and (iii) of Therorem $4.2$ 
of \cite{TUIOCOR} (taking into account also Theorem $4.3$ of the same paper). So we can refer
the reader to this paper for details. However, we can verify
directly that our solution
\eqref{3.12} satisfies eq. \eqref{0.1}. 
In order to do so, we use a slightly different notation. In particular,
let us write formula \eqref{3.12} as 
\begin{equation}\label{proof:1}
u(x;t)=-2 B^\dagger e^{-A^\dagger x} \zG^{-1}(x,t)e^{-A^\dagger x}e^{8 (A^\dagger)^3 t}C^\dagger\,,
\end{equation}
where 
$\zG(x,\,t)=I_p+Q(x,\,t)N(x)$ with $Q(x,\,t)=e^{-A^\dagger x}e^{8 (A^\dagger)^3 t}Q 
e^{-A x}e^{8 A^3 t}$ and $N(x)=e^{-A x}Ne^{-A^\dagger x}$,  $Q$ and $N$ being 
defined in \eqref{3.8}. We will see in Section \ref{sec:4} that these two matrices 
are the unique solutions of the so-called Lyapunov equations, i.e. eqs. \eqref{4.1a}. 
We recall the following rule: 
If $A(x)$ is an invertible matrix
of functions depending on $x$ such that its derivative with respect to $x$ exists, then
\begin{equation}\label{proof:2}
\frac{\partial}{\partial x}\left(A(x)^{-1}\right)=-A(x)^{-1}
\left(\frac{\partial}{\partial x}A(x)\right)A(x)^{-1}.
\end{equation}
Applying this differentiation rule to the function $\zG(x,\,t)$ and 
deriving \eqref{proof:1} with respect to $t$ , 
we easily get (from now on
we omit the dependence of $\zG(x,\,t)$ on $x$ and $t$)
\begin{equation}\label{proof:3}
u_t=-16B^\dagger e^{-A^\dagger x}\zG^{-1}\left[(A^\dagger)^3-Q(x,\,t)A^3 N(x)\right]
\zG^{-1}e^{-A^\dagger x}e^{8 (A^\dagger)^3 t}C^\dagger\,.
\end{equation}
Now, if one applies again \eqref{proof:2} to the function $\zG$ and take the $x$-derivative,
after some straightforward calculations and using also \eqref{4.1a}, we obtain
\begin{equation}\label{proof:4}
\left(\zG^{-1}\right)_x=\zG^{-1}A^\dagger+A^\dagger \zG^{-1}
-2\zG^{-1}(A^\dagger-QAN)\zG^{-1}\,.
\end{equation}
As a consequence of \eqref{proof:4}, we can calculate in a direct way $u_x$ obtaining
\begin{equation}\label{proof:5}
u_x=4B^\dagger e^{-A^\dagger x}\zG^{-1}\left[(A^\dagger)-Q(x,\,t)A N(x)\right]
\zG^{-1}e^{-A^\dagger x}e^{8 (A^\dagger)^3 t}C^\dagger\,.
\end{equation}
Now, computing the derivative of \eqref{proof:5} and taking into account \eqref{proof:4}
we find
\begin{align*}
 u_{xx}&=8 B^\dagger e^{-A^\dagger x}\zG^{-1}\left[(A^\dagger)^2+QA^2N-2(A^\dagger-QAN)
\zG^{-1}(A^\dagger-QAN)\right]\\ &\cdot
\zG^{-1}e^{-A^\dagger x}e^{8 (A^\dagger)^3 t}C^\dagger\,.
\end{align*}
By very similar calculations we get
\begin{align}\label{proof:6}
 \nonumber & u_{xxx}=16 B^\dagger e^{-A^\dagger x}\zG^{-1}
\Big[(A^\dagger)^3-QA^3N-3((A^\dagger)^2+QA^2N)\zG^{-1}(A^\dagger-QAN)
\\ \nonumber & -3(A^\dagger-QAN)\zG^{-1}((A^\dagger)^2+QA^2N)\\ & +6(A^\dagger-QAN)
\zG^{-1}(A^\dagger-QAN)\zG^{-1}(A^\dagger-QAN)\Big]
\zG^{-1}e^{-A^\dagger x}e^{8 (A^\dagger)^3 t}C^\dagger\,.
\end{align} 
Finally, using \eqref{proof:5} and taking into account that $u$ is a real scalar function, we obtain
\begin{align}\label{proof:7}
 \nonumber & 2|u|^2u_x=u_x u^\dagger u+uu^\dagger u_x=16 B^\dagger e^{-A^\dagger x}\zG^{-1}
\\ \nonumber & \cdot 
\Big[((A^\dagger)^2+QA^2N)\zG^{-1}(A^\dagger-QAN)
+(A^\dagger-QAN)\zG^{-1}((A^\dagger)^2+QA^2N)\\&-2(A^\dagger-QAN)
\zG^{-1}(A^\dagger-QAN)\zG^{-1}(A^\dagger-QAN)\Big]
\zG^{-1}e^{-A^\dagger x}e^{8 (A^\dagger)^3 t}C^\dagger\,,
\end{align}
and, at this point, it is very simple to observe that $u_{xxx}+6|u|^2u_x=-u_t$.
\end{proof}
Now, we will build a different explicit formula which is equivalent to the one expressed 
by \eqref{3.12}. In order to obtain this result, we first observe that $\zO(y,\,t)$,
as a consequence of the realness of the triplet $\left(A,\,B,\,C\right)$, is a real function. 
As a result, $$\zO(y;\,t)^\dagger=\zO(y;\,t).$$
By using this relation, eq. \eqref{2.3} can be written as
\begin{align}\label{3.17} 
\nonumber K(x,y;t)-\Big(C e^{-A x} & - \int_x^\infty dz\,\int_x^\infty ds \,
K(x,z;t)\,Ce^{-Az+8A^3 t}e^{-As}\,BC e^{-A^\dagger s}\Big)\\ 
& \cdot e^{-Ay+8(A)^3 t}B=0\,, \quad y>x\,.
\end{align}
By very similar computations we get the solution of \eqref{3.17} as 
$$K(x,y;t)=C E^{-1}(x,t)e^{-A^\dagger (y-x)}B\,,$$
where
\begin{align}\label{3.18}
E(x,\,t)=e^{2Ax-8A^3t}+Pe^{-2Ax+8A^3t}P,\quad P=\int_0^\infty ds\, e^{-As}BC e^{-As}\,.
\end{align}
Finally, using eq. \eqref{2.5} we obtain 
\begin{align}\label{3.19}
v(x;t)=-2 C E^{-1}(x,t)B\,.
\end{align}
The next theorem shows the relationships between the matrices $Q,\,N$ and $P$ 
and the solutions formula \eqref{3.12} and \eqref{3.19}. 
\begin{theorem}\label{th:3.2}
Suppose that the triplet $\left(A,\,B,\,C\right)$ is real and is a minimal 
representation of the kernel 
$\zO(y,t)$, and that the eigenvalues of the matrix $A$ have positive real parts. Then
\begin{enumerate}
\item[i.] The following relation holds: $NQ=P^2$.
\item[ii.] The matrix $P$ is invertible on the entire $xt$-plane.
\item[iii.] $E(x,\,t)=F(x,\,t)^\dagger$ and, as a consequence, the matrix $E(x,\,t)$ 
is invertible on the entire
$xt$-plane. 
\item[iv.] For each fixed $t$, $E(x,\,t)^{-1}\to 0$ as $x\to\pm \infty$ .
\item[v.] The real scalar function $v(x,t)$ satisfy \eqref{0.1} everywhere on the $xt$-plane.
Moreover, $v(x,t)$ is analytic on the entire $xt$-plane and tends to zero exponentially 
for each fixed $t$ as 
$x\to\pm \infty$.
\item[vi.] The explicit formulas \eqref{3.12} and \eqref{3.19} yield equivalent 
exact solutions to the mKdV eq. \eqref{0.1} 
everywhere on the entire xt-plane.
\end{enumerate}
\end{theorem}
We skip the proof of this theorem, because it can be constructed by rearranging
(with inessential modifications) the proofs 
of Theorems $5.2$, $5.4$ and $5.5$ of \cite{TUIOCOR}. 
However, we remark that:
\begin{itemize} 
\item By very similar calculations to those developed in the proof of Theorem \ref{th:3.1}, 
we can directly verify that the formula \eqref{3.19}
satisfies eq. \eqref{0.1} everywhere on the $xt$-plane.
\item Because $u(x,\,t)$ is real and scalar, 
the equivalence of \eqref{3.12} and \eqref{3.19} follows from the 
relation $E(x,\,t)=F(x,\,t)^\dagger$.
\end{itemize}
We conclude this section with the following 
\begin{theorem}\label{th:3.3}
Suppose that the triplet $\left(A,\,B,\,C\right)$ is real and is a minimal 
representation of the kernel 
$\zO(y,t)$, and that the eigenvalues of the matrix $A$ have positive real parts. Then
the solution to the mKdV equation
given in the equivalent forms \eqref{3.12} and \eqref{3.19} satisfies
\begin{align}\label{3.20}
[u_x(x,t)]^2
=\dfrac{\partial^2 \log(\det E(x,t))}{\partial x^2}=
\dfrac{\partial^2 \log(\det F(x,t))}{\partial x^2}.
\end{align}
\end{theorem}
The proof of this theorem can be found in \cite{TUIOCOR, AktDV7}. 

\section{Canonical form of the triplet $\left(A,\,B,\,C\right)$.}\label{sec:4}
In this section we show how it is possible, without loss of generality, to choose the triplet 
$\left(A,\,B,\,C\right)$ in ``canonical form."

To obtain this representation, we begin by finding the explicit solutions of the
mKdV given by \eqref{3.12} and \eqref{3.19} following an ``algorithmic" procedure. 
Starting with a real matrix triplet $\left(A,\,B,\,C\right)$ which realizes a minimal 
representation of the function $\zO(y;\,t)$ and such that the eigenvalues 
of $A$ have positive real 
parts, we consider the following equations:
\begin{subequations}\label{4.1}
\begin{align}
& A^\dagger Q+QA=C^\dagger C,\quad AN+NA^\dagger=B B^\dagger,\label{4.1a}\\
& AP+PA=BC\,. \label{4.1b}
\end{align}
\end{subequations}
Equations \eqref{4.1a} are the so-called Lyapunov equations, instead \eqref{4.1b} is
known as a Sylvester equation. These equations are studied in detail 
in \cite{Dym} where 
it is proved that, under our hypotheses on the triplet 
$\left(A,\,B,\,C\right)$, they are uniquely solvable. 
We have the following 
\begin{theorem}\label{th:4.1}
Suppose that the triplet $\left(A,\,B,\,C\right)$ is real and is a minimal representation of the kernel 
$\zO(y,t)$, and that the eigenvalues of the matrix $A$ have positive real parts. Then, the unique 
solutions $Q,\,N$ of the Lyapunov equations and the unique solution $P$ of the Sylvester equation
are such that
\begin{enumerate}
\item The matrices $Q,\,N$ and $P$ are real, and $Q$ and $N$ are selfadjoint.
\item The matrices $Q$ and $N$ can be expressed via \eqref{3.8}, instead 
the matrix $P$ is given by the formula $\eqref{3.18}$. Moreover, $Q$ and $N$ are 
simultaneously invertible, and also 
the matrix $P$ is invertible.
\end{enumerate}
\end{theorem}
A proof of this theorem can be found in \cite{TUIOCOR} and \cite{Dym}.

Suppose that we are able to solve \eqref{4.1a} (respectively, \eqref{4.1b}).
As a consequence of Theorem \ref{th:4.1}, their matrix solutions 
are those introduced in the preceding section by \eqref{3.8} (respectively, \eqref{3.18}).
 For these reasons, knowing
the solutions $Q$ and $N$ (respectively, $P$),
we can construct, in a unique way, the matrix $F(x,\,t)$ 
via eq. \eqref{3.11} (respectively, $E(x,\,t)$ given by \eqref{3.18}), and, finally, we can write down 
the solution of the mKdV by \eqref{3.12} (respectively, \eqref{3.19}). 

It is natural to look for a larger 
class including triplets such that the solutions of the corresponding Lyapunov or Sylvester equations   
have the same properties as in Theorem \ref{th:4.1}. In fact, for every triplet in this class,
we can repeat the procedure above introduced. For this reasons we introduce the following
\begin{definition}\label{def:4.1} 
We say that the triplet $\left(A,B,C\right)$ of size $p$ belongs
{\it to the admissible class} if the following conditions are met:
\item{$(i)$} The matrices $A,$ $B,$ and $C$ are real.
\item{$(ii)$} The triplet $\left(A,B,C\right)$ corresponds to the minimal realization 
when 
used in the right-hand side of (3.3).
\item{$(iii)$} None of the eigenvalues of
$A$ is purely imaginary and no two eigenvalues of
$A$ can occur symmetrically with respect to the imaginary axis in the complex
$\lambda$-plane.
\end{definition}
\begin{definition}\label{def:4.2} 
Two triplets $(\tilde{A},\tilde{B},\tilde{C})$ and $(A,B,C)$
are called {\it equivalent} if they lead to the same potential $u(x,t)$.
\end{definition}
Let us now consider a triplet $(\tilde{A},\tilde{B},\tilde{C})$ in the admissible class.
What can we state about the solutions of the corresponding Lyapunov or Sylvester equations
\begin{subequations}\label{4.2}
\begin{align}
& \tilde{A}^\dagger \tilde{Q}+\tilde{Q}\tilde{A}=\tilde{C}^\dagger \tilde{C},
\quad \tilde{A}\tilde{N}+\tilde{N}\tilde{A}^\dagger=\tilde{B} \tilde{B}^\dagger,\label{4.2a}\\
& \tilde{A}\tilde{P}+\tilde{P}\tilde{A}=\tilde{B}\tilde{C}\, ? \label{4.2b}
\end{align}
\end{subequations}
In \cite{TUIOCOR, Dym} the reader will find the proof of the following
\begin{theorem}\label{th:4.2}
If the triplet $(\tilde{A},\tilde{B},\tilde{C})$ belongs to the admissible class, 
then the following statements hold:
\begin{enumerate}
\item[1)] Equations \eqref{4.2a} and \eqref{4.2b} are uniquely solvable.
\item[2)] The matrix solutions $\tilde{Q}$ and $\tilde{N}$ of 
\eqref{4.2a} are selfadjoint.
Moreover, $\tilde{Q}$ and $\tilde{N}$ are simultaneously invertible and 
also $\tilde{P}$ is invertible.
\item[3)] The matrices  
\begin{subequations}
\begin{align}
& \tilde{F}(x,t)=e^{2\tilde{A}^\dagger x-8(\tilde{A}^\dagger)^3t}+\tilde{Q}
e^{-2\tilde{A}x+8\tilde{A}^3t}\tilde{N},\label{4.3a}\\
& \tilde{E}(x,\,t)=e^{2\tilde{A}x-8\tilde{A}^3t}+\tilde{P}
e^{-2\tilde{A}x+8\tilde{A}^3t}\tilde{P},\,\label{4.3b}
\end{align}
\end{subequations}
are invertible on the entire $xt$-plane.
\item[4)] The functions
\begin{equation}\label{4.4}
\tilde{u}(x;t)=-2 \tilde{B}^\dagger \tilde{F}^{-1}(x,t)\tilde{C}^\dagger,\,
\quad \tilde{v}(x;t)=-2 \tilde{C} \tilde{E}^{-1}(x,t)\tilde{B}\,.
\end{equation}
yield two different but equivalent explicit solutions of \eqref{0.1}. Moreover, they are 
analytic on the entire $xt$-plane and tend to zero exponentially for each fixed $t$ as 
$x\to\pm \infty$.
\end{enumerate}
\end{theorem}
Then for every triplet in the admissible class we can apply the algorithmic procedure above.

Now, a natural question is: Starting from a triplet $(\tilde{A},\tilde{B},\tilde{C})$
in the admissible class, is it possible to construct an equivalent triplet $(A,\,B,\,C)$ such that the matrices 
$A,\,B,\,C$ are real and are a minimal representation of the function $\zO(y,\,t)$, and the eigenvalues of $A$
have positive real parts (i.e., a triplet $(A,\,B,\,C)$ of type to this considered in Section \ref{sec:3})?
The answer to this question is affirmative. For the sake of space we refer the reader to
\cite{TUIOCOR} or \cite{AktDV7} where a complete solution to this problem is presented. In particular,
eqs. $(4.7)$ and $(4.8)$ of \cite{TUIOCOR} yield the triplet $(A,\,B,\,C)$ by starting from the triplet 
$(\tilde{A},\tilde{B},\tilde{C})$, while eqs. $(4.10)$, $(4.11)$ in the same paper explain  
how to construct $Q,\,N,\,E,\,F$ from $\tilde{Q},\,\tilde{N},\,\tilde{E},\,\tilde{F}$.
As a consequence, the following theorem allows us to understand which is the ``canonical way" 
to choose the triplet
$(A,\,B,\,C)$ in \eqref{3.4} and, consequently, in the explicit formulas \eqref{3.12} and \eqref{3.19}.
\begin{theorem}\label{th:4.3}
\it
To any admissible triplet $(\tilde A,\tilde B,\tilde C),$
one can associate a special admissible
triplet $(A,B,C),$ where $A$ has the Jordan canonical form
with each Jordan block containing a distinct eigenvalue having a positive real part,
the column $B$ consists of zeros
and ones, and $C$ has real entries. More specifically,
for some appropriate positive integer $m$ we have
\begin{align}\label{4.5}
A=\begin{pmatrix} A_1&
0&\cdots&0\\
0& A_2&\cdots&0\\
\vdots&\vdots &\ddots&\vdots\\
0&0&\cdots&A_m\end{pmatrix},
\qquad
B=\begin{pmatrix}
B_1\\
B_2\\
\vdots\\
B_m\end{pmatrix},\qquad
C=\begin{pmatrix} C_1&C_2& \cdots
&
C_m\end{pmatrix},
\end{align} 
where in the case of a real (positive) eigenvalue $\omega_j$
of $A_j$ the corresponding blocks are given by
\begin{align}\label{4.6}
C_j:=\begin{pmatrix} c_{jn_j}& \cdots
&
c_{j 2}&
c_{j 1}\end{pmatrix},
\end{align}
\begin{align}\label{4.7}
A_j:=\begin{pmatrix} \omega_j&
-1&0&\cdots&0&0\\
0& \omega_j& -1&\cdots&0&0\\
0&0&\omega_j& \cdots&0&0\\
\vdots&\vdots &\vdots &\ddots&\vdots&\vdots\\
0&0&0&\cdots&\omega_j&-1\\
0&0&0&\cdots&0&\omega_j\end{pmatrix},
\qquad
B_j:=\begin{pmatrix} 0\\
\vdots\\
0\\
1\end{pmatrix},
\end{align}
with $A_j$ having size $n_j\times n_j,$ $B_j$ size $n_j\times 1,$
$C_j$ size $1\times n_j,$  and
the constant $c_{jn_j}$ is nonzero.
In the case of complex eigenvalues, which must appear
in pairs as $\alpha_j\pm i\beta_j$ with $\alpha_j>0,$ the corresponding
blocks are given by
\begin{align}\label{4.8}
C_j:=\begin{pmatrix} \gamma_{j n_j}& \epsilon_{j n_j}&\dots
&
\gamma_{j 1}& \epsilon_{j 1}\end{pmatrix},
\end{align}
\begin{align}\label{4.9}
A_j:=\begin{pmatrix} \Lambda_j&
-I_2&0&\dots&0&0\\
0& \Lambda_j& -I_2&\dots&0&0\\
0&0&\Lambda_j&\dots&0&0\\
\vdots&\vdots &\vdots &\ddots&\vdots&\vdots\\
0&0&0&\dots&\Lambda_j&-I_2\\
0&0&0&\dots&0&\Lambda_j\end{pmatrix}
,\quad
B_j:=\begin{pmatrix} 0\\
\vdots\\
0\\
1\end{pmatrix},
\end{align}
where $\gamma_{js}$ and $\epsilon_{js}$ for
$s=1,\dots,n_j$ are
real constants with $(\gamma_{jn_j}^2+\epsilon_{jn_j}^2)>0,$ $I_2$ denotes the
$2\times 2$ unit matrix, each column vector
$B_j$ has $2n_j$ components, each $A_j$
has size $2n_j\times 2n_j,$
and the $2\times 2$
matrix $\Lambda_j$ is defined as
\begin{align}\label{4.10}
\Lambda_j:=\begin{pmatrix} \alpha_j&\beta_j
\\
\noalign{\medskip}
-\beta_j& \alpha_j\end{pmatrix}.
\end{align}
\end{theorem}
\begin{proof} 
The real triplet $(A,B,C)$ can be chosen as
in Section~3 of \cite{AktDV7}.
\end{proof}

\section{Significant examples}\label{sec:5}
{\bf Example 1}: Choosing the triplet $\left(A,\,B,\,C\right)$ as
\begin{equation*}
A=\begin{pmatrix}a\end{pmatrix},\quad B=\begin{pmatrix}1\end{pmatrix},\quad
C=\begin{pmatrix}c\end{pmatrix}
\end{equation*}
where $a>0$ and $0\neq c\in\R$ and solving the Sylvester eq. \eqref{4.1b}, we get
\begin{equation*}
P=\begin{pmatrix}\frac{c}{2a}\end{pmatrix}\,.
\end{equation*}
By using eq. \eqref{3.19}, we obtain
\begin{equation*}
v(x,\,t)=\dfrac{-2c}{e^{2ax-8a^3t}+\frac{c^2}{4a^2}e^{-2ax+8a^3t}}\,,
\end{equation*}
which may be called a ``single-soliton solution" to \eqref{0.1}. 

{\bf Example 2}: In this example we consider the case in which the transmission coefficients have a pole 
of order three. More precisely, let us consider the following triplet
\begin{equation*}
A=\begin{pmatrix}1& -1 & 0\\0 & 1 &-1\\0 & 0 & 1 \end{pmatrix},\quad B=\begin{pmatrix}0\\0\\1\end{pmatrix},\quad
C=\begin{pmatrix}1 & 2 & 1/2\end{pmatrix}.
\end{equation*}
It is not difficult to verify that the following matrix $P$ satisfies \eqref{4.1b}
\begin{equation*}
P=\begin{pmatrix} 1/8 & 7/16 & 5/8 \\1/4 & 3/4 & 13/16 \\1/2 & 5/4 & 7/8\end{pmatrix}\,.
\end{equation*}
In this case it is not a good idea to unzip the solution formula \eqref{3.19} 
in order to write its analytic expression because this representation take a lot of pages!
However, using Mathematica it is very easy to plot this solution. In the next 
figure four different graphs of $u(x,t)$ for four fixed values of $t$ 
($t=0,\, t=1/4,\, t=1/2$ and $t=3/4$) are given.
\begin{figure}[ht]
  \centering
    \includegraphics[width=12cm]{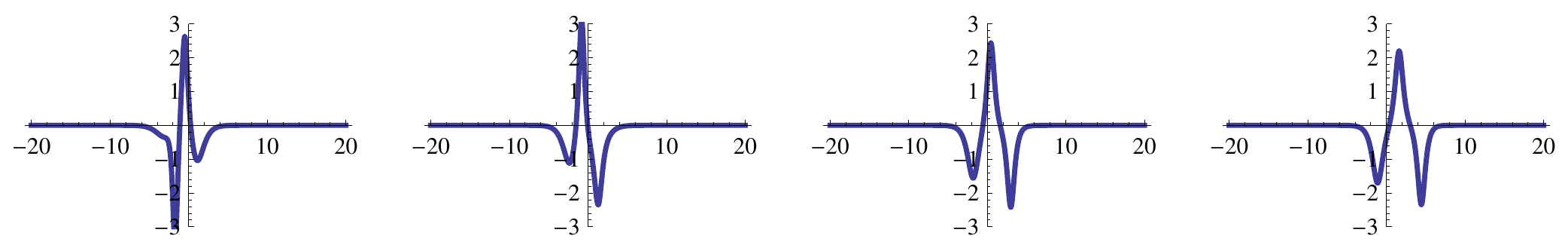}
\end{figure}  

\section*{Acknowledgements}

The author is greatly indebted to Cornelis van der Mee for useful discussions 
and to 
Antonio Aric\`o for his 
assistance in developing the Mathematica code.

\end{document}